\documentclass[preliminary,copyright,creativecommons]{eptcs}
\providecommand{\event}{EXPRESS/SOS 2017} 

\usepackage{amsmath,amssymb,amsthm}
\usepackage[all]{xy}

\newcommand{\Set}{\mathsf{Set}}
\newcommand{\Setc}{\mathsf{cSet}}
\newcommand{\Rel}{\mathsf{Rel}}
\newcommand{\DCPO}{\mathsf{DCPO}_\bot}

\newcommand{\PreOrd}{\mathsf{PreOrd}}
\newcommand{\alg}[1]{\mathsf{alg}(#1)}
\newcommand{\coalg}[1]{\mathsf{coalg}(#1)}
\newcommand{\cnt}{c}
\newcommand{\coalgc}[1]{\mathsf{coalg}_\cnt(#1)}

\newcommand{\kacc}[1]{[\C,\C]^\kappa}

\newcommand{\C}{\mathcal{C}}

\newcommand{\M}{\mathcal{M}}

\newcommand{\emalg}[1]{\mathsf{Alg}(#1)}
\newcommand{\emcoalg}[1]{\mathsf{CoAlg}(#1)}

\newcommand{\pow}{\mathcal{P}}
\newcommand{\powf}{\pow_f}
\newcommand{\powc}{\pow_c}
\newcommand{\Id}{\mathsf{Id}}

\newcommand{\real}{\mathbb{R}}

\newcommand{\lift}[1]{\overline{#1}}

\newcommand{\liftco}[1]{\overline{#1}}

\newcommand{\gr}{\mathsf{Graph}}

\newcommand{\id}{\mathsf{id}}

\newcommand{\op}{\mathsf{op}}

\newtheorem{theorem}{Theorem}
\newtheorem{lemma}{Lemma}
\newtheorem{corollary}{Corollary}

\theoremstyle{definition}
\newtheorem{definition}{Definition}
\newtheorem{example}{Example}

\title{Distributive Laws for Monotone Specifications\thanks{The research leading to these results has received funding from the
European Research Council under the European Union's Seventh Framework
Programme (FP7/2007-2013) / ERC grant agreement nr.~320571, and the Netherlands Organisation for Scientific Research (NWO), CoRE project, dossier number: 612.063.920. Part of this research was carried out during a visit of the author to the
University of Warsaw, supported by the Warsaw Center of Mathematics and Computer Science (WCMCS).}}
\author{Jurriaan Rot
\institute{Radboud University, Nijmegen}\\
\email{jrot@cs.ru.nl}
}
\def\titlerunning{Distributive Laws for Monotone Specifications}
\def\authorrunning{J. Rot}
\begin{document}
\maketitle

\begin{abstract}
Turi and Plotkin introduced an elegant approach to structural operational semantics based on universal coalgebra,	
parametric in the type of syntax and the type of behaviour.
Their framework includes abstract GSOS, a categorical generalisation of the classical GSOS rule format, as well as its 
categorical dual, coGSOS. Both formats are well behaved, in the sense that each specification has a unique model on which 
behavioural equivalence is a congruence. 
Unfortunately, the combination of the two formats does not feature these desirable properties. 
We show that \emph{monotone} specifications---that disallow negative premises---do induce a canonical distributive law 
of a monad over a comonad, and therefore a unique, compositional interpretation. 
\end{abstract}

\section{Introduction}\label{intro}

Structural operational semantics (SOS) is an expressive and popular framework for 
defining the operational semantics of programming languages and calculi.
There is a wide variety of specification formats that syntactically restrict
the full power of SOS, but guarantee certain desirable properties to hold~\cite{AFV}. 
A famous example is the so-called GSOS format~\cite{BloomIM95}. 
Any GSOS specification induces a unique interpretation which is compositional with 
respect to (strong) bisimilarity.

In their seminal paper~\cite{TP97}, Turi and Plotkin introduced an elegant mathematical approach to structural operational semantics,
where the type of syntax is modeled by an endofunctor $\Sigma$ and the type of behaviour is modeled by an endofunctor $B$. 
Operational semantics is then given by a \emph{distributive law} of $\Sigma$ over $B$.
In this context, models are \emph{bialgebras}, which consist of a $\Sigma$-algebra and a $B$-coalgebra 
over a common carrier. One major advantage of this framework over traditional approaches is that it is parametric 
in the type of behaviour.
Indeed, by instantiating the theory to a particular functor $B$, one can obtain well behaved specification formats for 
probabilistic and stochastic systems, weighted transition systems, streams, and many more~\cite{Klin09,Klin11,Bartels04}.

Turi and Plotkin introduced several kinds of natural transformations involving $\Sigma$ and $B$,
the most basic one being of the form $\Sigma B \Rightarrow B \Sigma$. If $B$ is a functor
representing labelled transition systems, then a typical rule that can be represented in this format 
is the following:
\begin{equation}\label{eq:simple-ex}
\frac{x \xrightarrow{a} x' \qquad y \xrightarrow{a} y'}{x \otimes y \xrightarrow{a} x' \otimes y'}
\end{equation}
This rule should be read as follows: if $x$ can make an $a$-transition to $x'$, and $y$ an $a$-transition to $y'$,
then $x \otimes y$ can make an $a$-transition to $x' \otimes y'$. 
Any specification of the above kind induces a unique \emph{supported model}, 
which is a $B$-coalgebra over the initial algebra of $\Sigma$.
If $\Sigma$ represents a signature and $B$ represents labelled transition systems, then this model  
is a transition system of which the state space is the set of closed terms in the signature, 
and, informally, a term makes a transition to another term if and only if there is a rule in the 
specification justifying this transition.

A more interesting kind is an \emph{abstract GSOS specification}, which is 
a natural transformation of the form $\Sigma (B \times \Id) \Rightarrow B\Sigma^*$, 
where $\Sigma^*$ is the \emph{free monad} for $\Sigma$ (assuming it exists). 
If $B$ is the functor that models (image-finite) transition systems, and $\Sigma$ is a functor representing
a signature, then such specifications correspond to 
classical GSOS specifications~\cite{TP97,Bartels04}.
As opposed to the basic format, GSOS rules allow complex terms in conclusions, as in the following rule specifying a constant $c$:
\begin{equation}\label{eq:gsos-ex}
\frac{}{c \xrightarrow{a} \sigma(c)}
\end{equation}
where $\sigma$ is some other operator in the signature (represented by $\Sigma$), which can itself be defined by some GSOS rules. 
The term $\sigma(c)$ is constructed from a constant and a unary operator from the signature,
as opposed to the conclusion $x' \otimes y'$ of the rule in~\eqref{eq:simple-ex}, which consists of a single operator and variables.
Indeed, the free monad $\Sigma^*$ occurring in an abstract GSOS specification is precisely what allows a complex term such as $\sigma(c)$
in the conclusion.
 
Dually, one can consider \emph{coGSOS specifications}, which
are of the form $\Sigma B^\infty \Rightarrow B(\Sigma + \Id)$, where $B^\infty$ is the cofree comonad for $B$ (assuming it exists).
In the case of image-finite labelled transition systems, this format corresponds to the \emph{safe ntree format}~\cite{TP97}.
A typical coGSOS rule is the following:
\begin{equation}\label{eq:cogsos-ex}
\frac{x \xrightarrow{a} x' \qquad x' \not \xrightarrow{a} }
{\sigma(x) \xrightarrow{a} x'}
\end{equation}
This rule uses two steps of lookahead in the premise; this is supported by the cofree comonad $B^\infty$ in the natural transformation.
The symbol $x' \not \xrightarrow{a}$ represents a \emph{negative premise}, which is satisfied whenever $x'$ does not make an $a$-transition.

Both GSOS and coGSOS specifications induce \emph{distributive laws}, and as a consequence they induce unique 
supported models on which behavioural equivalence is a congruence. 
The two formats are incomparable in terms of expressive power: GSOS specifications allow rules that involve complex 
terms in the conclusion, whereas coGSOS allows arbitrary lookahead in the arguments. 
It is straightforward to \emph{combine} GSOS and coGSOS as a natural transformation of the 
form $\Sigma B^\infty \Rightarrow B\Sigma^*$, called a \emph{biGSOS specification}, generalising both formats. 
However, such specifications are, in some sense, too expressive: they do not induce unique supported models,
as already observed in~\cite{TP97}. For example, 
the rules $\eqref{eq:gsos-ex}$ and $\eqref{eq:cogsos-ex}$ above (which are GSOS and coGSOS respectively) can be combined
into a single biGSOS specification. Suppose this combined specification has a model. By the axiom for $c$, there is a transition
$c \xrightarrow{a} \sigma(c)$ in this model. However, is there a transition $\sigma(c) \xrightarrow{a} \sigma(c)$? If 
there is not, then by the rule for $\sigma$, there is; but if there is such a transition, 
then it is not derivable, so it is not in the model! Thus, a supported model does not exist.
In fact, it was recently shown that, for biGSOS,
it is undecidable whether a (unique) supported model exists~\cite{KlinN17}.

The use of negative premises in the above example (and in~\cite{KlinN17}) is crucial. In the present paper, 
we introduce the notion of \emph{monotonicity} 
of biGSOS specifications, generalising monotone abstract GSOS~\cite{MS10}. In the case that $B$ is a functor representing labelled transition systems, this corresponds to 
the absence of negative premises, but the format does allow lookahead in premises as well as complex terms in conclusions. 
Monotonicity requires an \emph{order} on the functor $B$---technically, our definition of monotonicity is based on the \emph{similarity} order~\cite{HJ04} induced on the final coalgebra. 

We show that if there is a pointed DCPO structure on the functor $B$, then any monotone biGSOS specification yields a 
\emph{least} model as its operational interpretation. Indeed, monotone specifications do not necessarily have
a unique model, but it is the least model which makes sense operationally, since this corresponds to the natural notion that
every transition has a \emph{finite} proof. 
Our main result is that if the functor $B$ has a DCPO structure, then every monotone specification yields a canonical \emph{distributive law} of the free monad
for $\Sigma$ over the cofree comonad for $B$. Its unique model coincides with the least supported model of the specification. 
As a consequence, behavioural equivalence on this model is a congruence. 

However, the conditions of these results are a bit too restrictive: they rule out labelled transition systems,
the main example. The problem is that the functors typically used to model transition systems either
fail to have a cofree comonad (the powerset functor) or to have a DCPO structure (the finite or countable powerset functor).
In the final section, we mitigate this problem using the theory of (countably) presentable categories and accessible functors.
This allows us to relax the requirement of DCPO structure only to countable sets, given
that the functor $B$ is countably accessible (this is weaker than being finitary, a standard
condition in the theory of coalgebras) and the syntax consists only of countably many operations each with
finite arity. In particular, this applies to labelled transition systems (with countable branching) and
certain kinds of weighted transition systems.

\paragraph{Related work}
The idea of studying distributive laws of monads over comonads that are not induced by GSOS or coGSOS specifications
has been around for some time (e.g.,~\cite{Bartels04}), but, according to a recent overview paper~\cite{Klin11}, general bialgebraic formats 
(other than GSOS or coGSOS) which induce such distributive laws have not been proposed so far. In fact, it is shown
by Klin and Nachy\l{}a that the general problem of extending biGSOS specifications
to distributive laws is undecidable~\cite{KlinN14,KlinN17}. The current paper shows
that one does obtain distributive laws from biGSOS specifications when monotonicity is assumed (negative
premises are disallowed).
A fundamentally different approach to positive formats with lookahead,
not based on the framework of bialgebraic semantics but on labelled transition systems modeled very generally in a topos, was introduced in~\cite{Staton08}. It is deeply rooted in labelled transition systems, and hence seems incomparable to our approach
based on generic coalgebras for ordered functors. An abstract study of distributive laws 
of monads over comonads and possible morphisms between them is in~\cite{PowerW02}, but
it does not include characterisations in terms of simpler natural transformations.

\paragraph{Structure of the paper}
Section~\ref{sec:bialg} contains the necessary preliminaries on bialgebras and distributive laws. 
In Section~\ref{sec:similarity} we recall the notion of similarity on coalgebras, which
we use in Section~\ref{sec:mon-spec} to define monotone specifications and prove the existence of least supported models. 
Section~\ref{sec:dl} contains our main result: canonical distributive laws for monotone biGSOS specifications. 
In Section~\ref{sec:cpo}, this is extended to countably accessible functors.

\paragraph{Notation}
We use the categories $\Set$ of sets and functions,
$\PreOrd$ of preorders and monotone functions, 
and $\DCPO$ of pointed DCPOs and continuous maps. 
By $\pow$ we denote the (contravariant) power set functor; 
$\powc$ is the countable power set functor and $\powf$ the finite power set functor. 
Given a relation $R \subseteq X \times Y$, we write
$\pi_1 \colon R \rightarrow X$ and $\pi_2 \colon R \rightarrow Y$
for its left and right projection, respectively. 
Given another relation $S \subseteq Y \times Z$ we denote the composition
of $R$ and $S$ by $R \circ S$.
We let $R^\op = \{(y,x) \mid (x,y) \in R\}$.
For a set $X$, we let $\Delta_X = \{(x,x) \mid x \in X\}$. 
The graph of a function $f\colon X \rightarrow Y$ is $\gr(f) = \{(x,f(x)) \mid x \in X\}$.
The image of a set $S \subseteq X$ under $f$ is denoted simply by $f(S) = \{f(x) \mid x \in S\}$, and 
the inverse image of $V \subseteq Y$ by $f^{-1}(V) = \{x \mid f(x) \in V\}$.
The pairing of two functions $f,g$ with a common domain is denoted by $\langle f, g \rangle$
and the copairing (for functions $f,g$ with a common codomain) by $[f,g]$.
The set of functions from $X$ to $Y$ is denoted by $Y^X$.
Any relation $R \subseteq Y \times Y$ can be lifted pointwise
to a relation on $Y^X$; in the sequel we will simply denote such a pointwise
extension by the relation itself, i.e., for functions $f, g \colon X \rightarrow Y$
we have $f \, R \, g$ iff $f(x) \, R \, g(x)$ for all $x \in X$,
or, equivalently, $(f \times g)(\Delta_X) \subseteq R$.

\paragraph{Acknowledgements} The author is grateful to Henning Basold, Marcello Bonsangue, Bartek Klin and Beata Nachy\l{}a for inspiring discussions and suggestions.

\section{(Co)algebras, (co)monads and distributive laws}\label{sec:bialg}

We recall the necessary definitions on algebras, coalgebras, and distributive laws of monads over comonads.
For an introduction to coalgebra see~\cite{Rutten00,BR12}. All of the definitions and results below 
and most of the examples can be found in~\cite{Klin11}, which provides an overview of bialgebraic semantics.
Unless mentioned otherwise, all functors considered are endofunctors on $\Set$.

\subsection{Algebras and monads}

An \emph{algebra} for a functor $\Sigma \colon \Set \rightarrow \Set$ consists of a set $X$ and a function $f \colon \Sigma X \rightarrow X$.
An \emph{(algebra) homomorphism} from $f \colon \Sigma X \rightarrow X$ to $g \colon \Sigma Y \rightarrow Y$ is a function $h \colon X \rightarrow Y$
such that $h \circ f = g \circ \Sigma h$. The category of algebras and their homomorphisms is denoted by $\alg{\Sigma}$.

A \emph{monad} is a triple $\mathcal{T} = (T, \eta, \mu)$ where $T \colon \Set \rightarrow \Set$ is a functor and  
$\eta \colon \Id \Rightarrow T$ and $\mu \colon TT \Rightarrow T$ are natural transformations such
that $\mu \circ T \eta = \id = \mu \circ \eta T$ and $\mu \circ \mu T = \mu \circ T \mu$.
An \emph{(Eilenberg-Moore, or EM)-algebra} for $\mathcal{T}$ is a $T$-algebra 
$f \colon TX \rightarrow X$ such that $f \circ \eta_X = \id$ and $f \circ \mu_X = f \circ Tf$.
We denote the category of EM-algebras by $\emalg{\mathcal{T}}$.

We assume that a \emph{free monad} $(\Sigma^*, \eta, \mu)$ for $\Sigma$ exists.
This means that there is a natural transformation $\iota \colon \Sigma \Sigma^* \Rightarrow \Sigma^*$ such
that $\iota_X$ is a \emph{free algebra} on the set $X$ of generators, that is, the copairing of
$$
\xymatrix{
  \Sigma \Sigma^* X \ar[r]^-{\iota_X} & \Sigma^* X & X \ar[l]_-{\eta_X} 
}
$$
is an initial algebra for $\Sigma + X$.
By Lambek's lemma, $[\iota_X, \eta_X]$ is an isomorphism.
Any algebra $f \colon \Sigma X \rightarrow X$ induces a $\Sigma + X$-algebra $[f, \id]$,
and therefore by initiality a $\Sigma^*$-algebra $f^* \colon \Sigma^* X \rightarrow X$,
which we call the \emph{inductive extension} of $f$. In particular, the inductive extension of
$\iota_X$ is $\mu_X$. This construction preserves homomorphisms: if $h$ is a homomorphism from $f$ to $g$,
then it is also a homomorphism from $f^*$ to $g^*$.

\begin{example}
	An algebraic signature (a countable collection of operator names with finite arities) induces a 
	polynomial functor $\Sigma$, meaning here
	a countable coproduct of finite products. The free monad $\Sigma^*$ constructs terms, that is, 
	$\Sigma^*X$ is given by the grammar $t ::= \sigma(t_1, \ldots, t_n) \mid x$ where $x$ ranges over $X$ and
	$\sigma$ ranges over the operator names (and $n$ is the arity of $\sigma$), so in particular $\Sigma^*\emptyset$
	is the set of closed terms over $\Sigma$.
\end{example}

\subsection{Coalgebras and comonads}

A \emph{coalgebra} for the functor $B$ consists of a set $X$ and a function $f \colon X \rightarrow BX$.
A \emph{(coalgebra) homomorphism} from $f \colon X \rightarrow BX$ to $g \colon Y \rightarrow BY$ is a function
$h \colon X \rightarrow Y$ such that $Bh \circ f = g \circ h$. The category of $B$-coalgebras
and their homomorphisms is denoted by $\coalg{B}$. 

A \emph{comonad} is a triple $\mathcal{D} = (D, \epsilon, \delta)$ consisting of a functor $D \colon \Set \rightarrow \Set$
and natural transformations $\epsilon \colon D \Rightarrow \Id$ and $\delta \colon D \Rightarrow DD$ 
satisfying axioms dual to the monad axioms.
The category of Eilenberg-Moore coalgebras for $\mathcal{D}$, defined dually to EM-algebras, is denoted by $\emcoalg{\mathcal{D}}$.

We assume that a \emph{cofree comonad} $(B^\infty, \delta, \epsilon)$ for $B$ exists.
This means that there is a natural transformation $\theta \colon B^\infty \Rightarrow BB^\infty$ such
that $\theta_X$ is a \emph{cofree coalgebra} on the set $X$, that is, the pairing of 
$$
\xymatrix{
  B B^\infty X & B^\infty X \ar[r]^{\epsilon_X} \ar[l]_-{\theta_X} & X
}
$$
is a final coalgebra for $B \times X$.
Any coalgebra $f \colon X \rightarrow BX$ induces a $B \times X$-coalgebra $ \langle f, \id \rangle$,
and therefore by finality a $B^\infty$-coalgebra $f^\infty \colon X \rightarrow B^\infty X$,
which we call the \emph{coinductive extension} of $f$. In particular, the coinductive extension of
$\theta_X$ is $\delta_X$. This construction preserves homomorphisms: if $h$ is a homomorphism from $f$ to $g$,
then it is also a homomorphism from $f^\infty$ to $g^\infty$.

\begin{example}\label{ex:comonad:streams}
	Consider the $\Set$ functor $BX = A \times X$ for a fixed set $A$. Coalgebras for $B$ are called \emph{stream systems}. 
	There exists a final $B$-coalgebra, whose carrier can be presented as the set 
	$A^\omega$ of all streams over $A$, i.e., $A^\omega = \{\sigma \mid \sigma \colon \omega \rightarrow A \}$
	where $\omega$ is the set of natural numbers. For a set $X$, $B^\infty X = (A \times X)^\omega$. 
	Given $f \colon X \rightarrow A \times X$, its coinductive extension $f^\infty \colon X \rightarrow B^\infty X$
	maps a state $x \in X$ to its infinite unfolding.
	The final coalgebra of $GX = A \times X + 1$ consists of finite and infinite streams over $A$, that is, 
	elements of $A^* \cup A^\omega$. For a set $X$, $G^\infty X = (A \times X)^\omega \cup (A \times X)^* \times X$. 
\end{example}
\begin{example}
	Labelled transition systems are coalgebras for the functor $(\pow -)^A$, where $A$ is a fixed set of labels.
	Image-finite transition systems are coalgebras for the functor $(\powf -)^A$, and 
	coalgebras for $(\powc -)^A$ are transition systems which have, for every action $a \in A$ and every state $x$,
	a countable set of outgoing $a$-transitions from $x$. A final coalgebra for $(\pow -)^A$ does not exist (so there is no cofree
	comonad for it). However, both $(\powf -)^A$ and $(\powc -)^A$
	have a final coalgebra, consisting of possibly infinite rooted trees, edge-labelled in $A$, modulo strong bisimilarity,
	where for each label, the set of children is finite respectively countable. 
	The cofree comonad of $(\powf -)^A$ respectively $(\powc -)^A$, applied to a set $X$,
	consist of all trees as above, node-labelled in $X$. 
\end{example}
	
	\begin{example} \label{ex:comonad:wts}
	A \emph{complete monoid} is a (necessarily commutative) monoid $M$ together with an infinitary sum
	operation consistent with the finite sum~\cite{droste2009semirings}.
	Define the functor $\M \colon \Set \rightarrow \Set$ by $\M(X) = \{\varphi \mid \varphi \colon X \rightarrow M\}$
	and, for $f \colon X \rightarrow Y$, $\M(h)(\varphi) = \lambda y. \sum_{x \in f^{-1}(y)} \varphi(x)$. 
	A \emph{weighted transition system} over a set of labels $A$ is a coalgebra $f \colon X \rightarrow (\M X)^A$. 
	Similar to the case of labelled transition systems, we obtain weighted
	transition systems whose branching is countable for each label as coalgebras for the
	functor $(\M_c -)^A$, where
	$\M_c$ is defined by $\M_c(X) = \{\varphi \colon X \rightarrow M \mid \varphi(x) \neq 0 \text{ for countably many }x \in X\}$. We note that this only requires a countable sum on $M$ to be well-defined and,
	by further restricting to finite support, weighted transition systems are defined for any 
	commutative monoid (see, e.g.,~\cite{Klin09}). 
	Labelled transition systems are retrieved by taking the monoid with two elements
	and logical disjunction as sum. Another example arises by taking the monoid $M = \real^+ \cup \{\infty\}$ of 
	non-negative reals extended with a top element $\infty$, with the supremum operation. 
\end{example}


\subsection{GSOS, coGSOS and distributive laws}

Given a signature, a \emph{GSOS} rule~\cite{BloomIM95} $\sigma$
of arity $n$ is of the form
  \begin{equation}\label{eq:gsos}
    \frac{\{x_{i_j} \stackrel{a_j}{\rightarrow} y_j\}_{j = 1..m} \qquad \{x_{i_k} \stackrel{b_k}{\not \rightarrow}\}_{k = 1..l}}
    {\sigma(x_1, \ldots, x_n) \stackrel{c}{\rightarrow} t}
  \end{equation}
where $m$ and $l$ are the number of positive and negative premises respectively;
$a_1, \ldots, a_m, b_1, \ldots, b_l, c \in A$ are labels; $x_1, \ldots, x_n$,
$y_1, \ldots, y_m$ are pairwise distinct variables, and $t$ is a term over these variables.
An \emph{abstract GSOS specification} is a natural transformation of the form 
$$
  \Sigma (B \times \Id) \Rightarrow B\Sigma^* \, .
$$
As first observed in~\cite{TP97}, specifications in the GSOS format are generalised by abstract GSOS specifications, 
where $\Sigma$ models the signature and $BX = (\powf X)^A$.

A \emph{safe ntree} rule (as taken from~\cite{Klin11}) for $\sigma$ is of the form 
$
  \frac{\{z_i \stackrel{a_i}{\rightarrow} y_i\}_{i \in I} \qquad \{w_j \stackrel{b_j}{\not \rightarrow}\}_{j \in J}}
  {\sigma(x_1, \ldots, x_n) \stackrel{c}{\rightarrow} t}
$
where $I$ and $J$ are countable possibly infinite sets, the $z_i$, $y_i$, $w_j$, $x_k$ are variables, and $b_j, c, a_i \in A$;
the $x_k$ and $y_i$ are all distinct and they are the only variables that occur in the rule; 
the dependency graph of premise variables (where positive premises are seen as
directed edges) is well-founded, and $t$ is either a variable or a term built of a single operator from
the signature and the variables.
A \emph{coGSOS specification} is a natural transformation of the form 
$$
  \Sigma B^\infty \Rightarrow B(\Sigma + \Id) \, .
$$
As stated in~\cite{TP97}, every safe ntree specification induces a coGSOS specification
where $\Sigma$ models the signature and $BX = (\powf X)^A$. 

A \emph{distributive law} of a monad $\mathcal{T} = (T, \eta, \mu)$ over a comonad $\mathcal{D} = (D, \epsilon, \delta)$ is a natural transformation 
$\lambda \colon TD \Rightarrow DT$ so that 
$\lambda \circ D \eta = \eta D$, $\epsilon T \circ \lambda = T\epsilon$,
$\lambda \circ \mu T = D \mu \circ \lambda T \circ T \lambda$
and $D \lambda \circ \lambda D \circ T \delta = \delta T \circ \lambda$.
A \emph{$\lambda$-bialgebra} is a triple $(X, f, g)$ 
where $X$ is a set, $f$ is an EM-algebra for $\mathcal{T}$ and $g$ is an EM-coalgebra for $\mathcal{D}$,
such that $g \circ f = Df \circ \lambda_X \circ Tg$.

Every distributive law $\lambda$ induces, by initiality, a unique coalgebra $h \colon T\emptyset \rightarrow DT\emptyset$ 
such that $(T\emptyset, \mu_\emptyset, h)$ is $\lambda$-bialgebra.
If $\mathcal{D}$ is the cofree comonad for $B$, then $h$ is the coinductive extension of a $B$-coalgebra $m \colon T\emptyset \rightarrow BT\emptyset$,
which we call the \emph{operational model} of $\lambda$. Behavioural equivalence
on the operational model is a congruence. 
This result applies in particular to abstract GSOS and coGSOS specifications, which
both extend to distributive laws of monad over comonad. 


A \emph{lifting} of a functor $T \colon \Set \rightarrow \Set$ to $\emcoalg{\mathcal{D}}$ is a functor 
$\lift{T}$ making the following commute:
$$
\xymatrix{
  \emcoalg{\mathcal{D}} \ar[d] \ar[r]^{\lift{T}} & \emcoalg{\mathcal{D}} \ar[d] \\
  \Set \ar[r]^{T} & \Set
}
$$
where the vertical arrows represent the forgetful functor, sending a coalgebra to its carrier.
Further, a monad $(\lift{T}, \lift{\eta}, \lift{\mu})$ on
$\emcoalg{\mathcal{D}}$ is a lifting of a monad $\mathcal{T} = (T,\eta,\mu)$ on $\Set$ if $\lift{T}$ is a lifting of $T$, $U\lift{\eta} = \eta U$ and $U \lift{\mu}= \mu U$.
A lifting of $\mathcal{T}$ to $\coalg{B}$ is defined similarly.

Distributive laws of $\mathcal{T}$ over $\mathcal{D}$ are in
one-to-one correspondence with liftings of $(T, \eta, \mu)$ to $\emcoalg{\mathcal{D}}$ (see~\cite{J75,TP97}).
If $\mathcal{D}$ is the cofree comonad for $B$, then
$\emcoalg{\mathcal{D}} \cong \coalg{B}$, hence a further equivalent
condition is that $\mathcal{T}$ lifts to $\coalg{B}$. 
In that case, the operational model of a distributive law can be
retrieved by applying the corresponding lifting to the unique 
coalgebra $! \colon \emptyset \rightarrow B\emptyset$.

\section{Similarity}\label{sec:similarity}

In this section, we recall the notion of \emph{simulations} of coalgebras from~\cite{HJ04}, and prove a few basic results
concerning the similarity preorder on final coalgebras.

An \emph{ordered functor} is a pair $(B, \sqsubseteq)$ of functors $B \colon \Set \rightarrow \Set$ 
and $\sqsubseteq \colon \Set \rightarrow \PreOrd$ such that
$$
\xymatrix{
		& \PreOrd \ar[d]\\
	\Set \ar[r]^{B} \ar[ur]^{\sqsubseteq} & \Set
}
$$
commutes, where the arrow from $\PreOrd$ to $\Set$ is the forgetful functor. Thus, given an ordered functor,
there is a preorder $\sqsubseteq_{BX} \subseteq BX \times BX$ for any set $X$, 
and for any map $f \colon X \rightarrow Y$, $Bf$ is monotone.

The (canonical) \emph{relation lifting} of $B$  is defined on a relation $R \subseteq X \times Y$ by 
$$
	\Rel(B)(R) = \{(b,c) \in BX \times BY \mid \exists d \in BR.\, B\pi_1(d) = b \text{ and } B\pi_2(d) = c \} \,.
$$
For a detailed account of relation lifting, see, e.g.,~\cite{Jacobs:coalg}.
Let $(B, \sqsubseteq)$ be an ordered functor. The \emph{lax relation lifting} $\Rel_{\sqsubseteq}$ is
defined as follows:
$$
	\Rel_{\sqsubseteq}(B)(R \subseteq X \times Y) = {\sqsubseteq_{BX}} \circ \Rel(B)(R) \circ {\sqsubseteq_{BY}} \, .
$$
Let $(X,f)$ and $(Y,g)$ be $B$-coalgebras. A relation $R \subseteq X \times Y$ is 
a \emph{simulation} (between $f$ and $g$) if $R \subseteq (f \times g)^{-1}(\Rel_{\sqsubseteq}(B)(R))$. 
The greatest simulation between coalgebras $f$ and $g$ is called \emph{similarity}, 
denoted by $\lesssim_f^g$, or $\lesssim_f$ if $f=g$, or simply $\lesssim$ if $f$ and $g$
are clear from the context.

Given a set $X$ and an ordered functor $(B, \sqsubseteq)$, we define the ordered functor $(B \times X, \widetilde{\sqsubseteq})$
by 
$$
  (b,x) \widetilde{\sqsubseteq}_{BX} (c,y) \quad \text{ iff } \quad b \sqsubseteq_{BX} c \text{ and }x=y \,.
$$ 
The induced notion of simulation can naturally be expressed
in terms of the original one:
\begin{lemma}\label{lm:order-id}
  Let $\lesssim$ be the similarity relation between coalgebras $\langle f, f' \rangle \colon X \rightarrow BX \times Z$ and
  $\langle g, g' \rangle \colon X \rightarrow BX \times Z$. Then for any relation $R \subseteq X \times X$, 
  we have $R \subseteq (\langle f, f'\rangle \times \langle g, g' \rangle)^{-1}(\Rel_{\widetilde{\sqsubseteq}}(B \times Z)(R))$
  iff $R \subseteq (f\times g)^{-1}(\Rel_{\sqsubseteq}(B)(R))$ and for all $(x,y) \in R$: $f'(x) = g'(x)$. 
\end{lemma}
Given an ordered functor $(B, \sqsubseteq)$
we write 
$$
\lesssim_{B^\infty X}
$$
for the similarity order induced by $(B \times X, \widetilde{\sqsubseteq})$ on the cofree coalgebra 
$(B^\infty X, \langle \theta_X, \epsilon_X \rangle)$. 
We discuss a few examples of ordered functors and similarity---see~\cite{HJ04}
for many more. 
\begin{example}\label{ex:ordered-functors-pow}
  For the functor $L_fX = (\powf X)^A$ ordered by (pointwise) subset inclusion, 
  a simulation as defined above is a (strong) simulation in the standard sense. For
  elements $p,q \in L_f^\infty X$, we have $p \lesssim_{L_f^\infty X} q$ iff there exists a (strong)
  simulation between the underlying trees of $p$ and $q$, so that related pairs agree on labels in $X$. 
\end{example}
\begin{example}\label{ex:ordered-functors-part}
  For any $G \colon \Set \rightarrow \Set$, the functor $B = G + 1$, where $1 = \{\bot\}$, can be ordered as follows: 
  $x \leq y$ iff $x = \bot$ or $x=y$, for all $x,y \in BX$.
  If $G = A \times \Id$ then $B^\infty X$ consists of 
  finite and infinite sequences of the form $x_0 \xrightarrow{a_0} x_1 \xrightarrow{a_1} x_2 \xrightarrow{a_2} \ldots$
  with $x_i \in X$ and $a_i \in A$ for each $i$ (cf. Example~\ref{ex:comonad:streams}).
  For $\sigma,\tau \in B^\infty X$ we have $\sigma \lesssim_{B^\infty X} \tau$
  if $\sigma$ does not terminate before $\tau$ does, and $\sigma$ and $\tau$ agree on labels in $X$ and $A$ on each 
  position where $\sigma$ is defined. 
\end{example}

\begin{lemma}\label{lm:sim-pres-hom}
	Coalgebra homomorphisms $h,k$ preserve
	similarity: if $x \lesssim y$ then $h(x) \lesssim k(y)$.
\end{lemma}

In the remainder of this section we state a few technical properties concerning similarity
on cofree comonads, which will be necessary in the following sections. 
The proofs use Lemma~\ref{lm:sim-pres-hom} and a few basic, standard properties of relation lifting.

Pointwise inequality of coalgebras implies pointwise similarity of coinductive extensions:

\begin{lemma}\label{lm:hom-sim}
	Let $(B, \sqsubseteq)$ be an ordered functor, and 
	let $f$ and $g$ be $B$-coalgebras on a common carrier $X$. If 
	$(f \times g)(\Delta_X) \subseteq {\sqsubseteq_{BX}}$ then 
	$(f^\infty \times g^\infty)(\Delta_X) \subseteq {\lesssim_{B^\infty X}}$.
\end{lemma}
Recall from Section~\ref{sec:bialg} that any $B$-homomorphism yields a $B^\infty$-homomorphism between coinductive extensions. 
A similar fact holds for inequalities.
\begin{lemma}\label{lm:ineq-comonad}
  Let $(B, \sqsubseteq)$ be an ordered functor where $B$ preserves weak pullbacks, and
	let $f \colon X \rightarrow BX$, $g \colon Y \rightarrow BY$ and $h \colon X \rightarrow Y$.
	\begin{itemize}
		\item If $Bh \circ f \sqsubseteq_{BY} g \circ h$ then $B^\infty h \circ f^\infty \lesssim_{B^\infty Y} g^\infty \circ h$, and conversely,
		\item if $Bh \circ f \sqsupseteq_{BY} g \circ h$ then $B^\infty h \circ f^\infty \gtrsim_{B^\infty Y} g^\infty \circ h$.
	\end{itemize}
\end{lemma}

\section{Monotone biGSOS specifications}\label{sec:mon-spec}

As discussed in the introduction, GSOS and coGSOS have a straightforward common generalisation,
called \emph{biGSOS} specifications.
Throughout this section we assume $(B, \sqsubseteq)$ is an ordered functor, $B$ has a cofree
comonad and $\Sigma$ has a free monad.
\begin{definition}\label{def:bigsos}
	A \emph{biGSOS specification} is a natural transformation of the form 
	$
		\rho \colon \Sigma B^\infty \Rightarrow B \Sigma^* 
	$.
  A triple $(X, a, f)$ consisting of a set $X$, an algebra $a \colon \Sigma X \rightarrow X$
  and a coalgebra $f \colon X \rightarrow BX$ (i.e., a bialgebra) is called  a \emph{$\rho$-model} if the following diagram commutes:
  $$
  \xymatrix{
  	\Sigma X \ar[rr]^a \ar[d]_{\Sigma f^\infty} 
  		&  
  		& X \ar[d]^f \\
  	\Sigma B^\infty X \ar[r]^{\rho_X} 
  		& B\Sigma^* X \ar[r]^{Ba^*}
  		& BX 
  }
  $$
\end{definition}
If $BX = (\powf X)^A$, then one can obtain biGSOS specifications from concrete rules in the \emph{ntree} format, 
which combines GSOS and safe ntree, allowing lookahead in premises, 
negative premises and complex terms in conclusions. 

Of particular interest are $\rho$-models on the initial algebra $\iota_\emptyset \colon 
\Sigma \Sigma^* \emptyset \rightarrow \Sigma^* \emptyset$: 
  \begin{equation}\label{eq:supp}
  \xymatrix{
  	\Sigma \Sigma^* \emptyset \ar[rr]^{\iota_\emptyset} \ar[d]_{\Sigma f^\infty} 
  		&  
  		& \Sigma^* \emptyset \ar[d]^f \\
  	\Sigma B^\infty \Sigma^* \emptyset \ar[r]^{\rho_{\Sigma^* \emptyset}} 
  		& B\Sigma^* \Sigma^* \emptyset \ar[r]^{B{\mu_\emptyset}}
  		& B\Sigma^* \emptyset
  }
  \end{equation}
(Notice that $\iota_\emptyset^* = \mu_\emptyset$.) We call
these \emph{supported models}. Indeed, for labelled transition systems, this notion coincides
with the standard notion of the supported model of an SOS specification (e.g.,~\cite{AFV}).

In the introduction, we have seen that biGSOS specifications do not necessarily induce
a supported model. But even if they do, such a model is not necessarily unique, and behavioural equivalence is 
not even a congruence, in general, as shown by the following example. 
\begin{example}\label{ex:more-models}
	In this example we consider a signature with 
	constants $c$ and $d$, and unary operators $\sigma$ and $\tau$. Consider the specification
	(represented by concrete rules)
	on labelled transition systems
	where $c$ and $d$ are not assigned any behaviour, 
	and $\sigma$ and $\tau$ are given by the following rules:
	$$
		\frac{x \xrightarrow{a} x' \qquad x' \xrightarrow{a} x''}
		{\sigma(x) \xrightarrow{a} x''}
		\qquad
		\frac{}
		{\tau(x) \xrightarrow{a} \sigma(\tau(x))}
	$$
	The behaviour of $\tau(x)$ is independent of its argument $x$. Which transitions can occur in a supported model? 
	First, for any $t$ there is a transition $\tau(t) \xrightarrow{a} \sigma(\tau(t))$. Moreover, 
	a transition $\sigma(\tau(t)) \xrightarrow{a} t''$ \emph{can} be in the model, although it does not need to be.
	But if it is there, it is supported by an \emph{infinite} proof. 
	
	In fact, one can easily construct a model in which the behaviour of $\sigma(\tau(c))$ is different from that of 
	$\sigma(\tau(d))$---for example, a model where $\sigma(\tau(c))$ does not make any transitions, whereas
	$\sigma(\tau(d)) \xrightarrow{a} t$ for some $t$. Then behavioural equivalence is not a congruence; 
	$c$ is bisimilar to $d$, but $\sigma(\tau(c))$ is not bisimilar to $\sigma(\tau(d))$.
\end{example}
The above example features a specification that has many different interpretations as a supported model.
However, there is only one which makes sense: the \emph{least} model, which only features
\emph{finite} proofs. It is sensible to speak
about the least model of this specification, since it does not contain any negative premises.
More generally, absence of negative premises can be defined based on an ordered functor and the induced similarity
order.
\begin{definition}
	A biGSOS specification $\rho \colon \Sigma B^\infty \Rightarrow B \Sigma^*$ is \emph{monotone}
	if the restriction of $\rho_X \times \rho_X$ to $\Rel(\Sigma)(\lesssim_{B^\infty X})$
	corestricts to $\sqsubseteq_{B\Sigma^*X}$,
	for any set $X$.
\end{definition}
If $\Sigma$ represents an algebraic signature, then monotonicity can be conveniently restated as follows (c.f.~\cite{BonchiPPR17}, where monotone GSOS is characterised in a similar way). For
every operator $\sigma$: 
$$
\frac{b_1 \lesssim_{B^\infty X} c_1 \quad\ldots\quad b_n \lesssim_{B^\infty X} c_n}
 {\rho_X(\sigma(b_1, \ldots, b_n)) \sqsubseteq_{B \Sigma^* X} \rho_X(\sigma(c_1, \ldots, c_n))}
$$
for every set $X$ and every $b_1, \ldots, b_n, c_1, \ldots, c_n \in B^\infty X$. Thus, in a monotone 
specification, if $c_i$ simulates $b_i$ for each $i$, then the behaviour of $\sigma(b_1, \ldots, b_n)$ is ``less than'' 
the behaviour of $\sigma(c_1, \ldots, c_n)$.

In the case of labelled transition systems, it is straightforward that monotonicity rules out (non-trivial use of) negative premises.
Notice that the example specification in the introduction consisting of rules~\eqref{eq:gsos-ex} and~\eqref{eq:cogsos-ex}, 
which does not have a model, is not monotone.
This is no coincidence: every monotone biGSOS specification has a model, if $B\Sigma^*\emptyset$ is a pointed DCPO, as we will see next. 
In fact, the proper canonical choice is the \emph{least} model, corresponding to behaviour obtained in finitely many proof steps. 

\subsection{Models of monotone specifications}

Let $\rho$ be a monotone biGSOS specification.
Suppose $B\Sigma^*\emptyset$ is a pointed DCPO. Then the set of coalgebras
$\coalg{B}_{\Sigma^*\emptyset} = \{f \mid f \colon \Sigma^*\emptyset \rightarrow B\Sigma^*\emptyset \}$, ordered pointwise, is a pointed DCPO as well.

Consider the function $\varphi \colon \coalg{B}_{\Sigma^*\emptyset} \rightarrow \coalg{B}_{\Sigma^*\emptyset}$, defined as follows: 
\begin{equation}\label{eq:phi}
	\varphi(f) = B\mu_\emptyset \circ \rho_{\Sigma^*\emptyset} \circ \Sigma f^\infty \circ \iota_\emptyset^{-1}
\end{equation}
Since $\iota_\emptyset$ is an isomorphism, a function $f$ is a fixed point of $\varphi$ if and only if it is a supported
model of $\rho$ (Equation~\eqref{eq:supp}).
  We are interested in the \emph{least} supported model. To show that
  it exists, since $\coalg{B}_{\Sigma^*\emptyset}$ is a 
  pointed DCPO, it suffices to show that $\varphi$ is monotone.
  

\begin{lemma}\label{lm:phi-mon}
	The function $\varphi$ is monotone.
\end{lemma}
\begin{proof}
  Suppose $f,g \colon \Sigma^*\emptyset \rightarrow B\Sigma^*\emptyset$ and $f \sqsubseteq_{B\Sigma^*\emptyset} g$. 
  By Lemma~\ref{lm:hom-sim}, we have $f^\infty \lesssim_{B^\infty \Sigma^* \emptyset} g^\infty$.
  From standard properties of relation lifting we derive
  $\Sigma f^\infty \mathrel{\Rel(\Sigma)(\lesssim_{B^\infty \Sigma^* \emptyset})} \Sigma g^\infty$
  and now
  the result follows by monotonicity of $\rho$ (assumption) 
  and monotonicity of $B \mu_\emptyset$ ($B$ is ordered).
\end{proof}
\begin{corollary}\label{cor:model}
If $B\Sigma^*\emptyset$ is a pointed DCPO and $\rho$ is a monotone biGSOS specification, then $\rho$ has a least supported model.
\end{corollary}
The condition of the Corollary is satisfied if $B$ is of the
form $B = G + 1$ (c.f. Example~\ref{ex:ordered-functors-part}), that is, 
$B = G + 1$ for some functor $G$ (where the element in the singleton 1
is interpreted as the least element of the pointed DCPO). 
Consider, as an example, the functor $BX = A \times X + 1$ of finite and infinite
streams over $A$. Any specification that does not mention termination (i.e., a specification
for the functor $GX = A \times X$) yields a monotone specification for $B$. 
\begin{example}
	Consider the following specification (in terms of rules) for the functor $BX = \mathbb{N} \times X + 1$ of (possibly terminating) stream systems
	over the natural numbers.
	It specifies a unary operator $\sigma$, a binary operator $\oplus$, 
	infinitely many unary operators $m \otimes -$ (one for each $m \in \mathbb{N}$), and
	constants $\mathit{ones}, \mathit{pos}$, $c$:
 	$$
 	\frac{x \xrightarrow{n} x' \qquad x' \xrightarrow{m} x''}
 	{\sigma(x) \xrightarrow{n} n \otimes (m \otimes \sigma(x''))}
 	\qquad 
 	\frac{x \xrightarrow{n} x' \qquad y \xrightarrow{m} y'}
 	{x \oplus y \xrightarrow{n+m} x' \oplus y'}
 	\qquad 
 	\frac{x \xrightarrow{n} x'}
 	{m \otimes x \xrightarrow{m \times n} m \otimes x'} 
 	$$
 	$$
 	\frac{}{\mathit{ones} \xrightarrow{1} \mathit{ones}}
 	\qquad
 	\frac{}{\mathit{pos} \xrightarrow{1} \mathit{ones} \oplus \mathit{pos}}
 	\qquad 
 	\frac{}{c \xrightarrow{1} \sigma(c)}
 	$$
 	where $+$ and $\times$ denote addition and multiplication of natural numbers, respectively.
 	This induces a monotone biGSOS specification; the rule for $\sigma$ is GSOS nor coGSOS, since
 	it uses both lookahead and a complex conclusion. By the above Corollary, it has a model.
 	The coinductive extension maps $\mathit{pos}$
 	to the increasing stream of positive integers, and $\sigma(\mathit{pos})$
 	is the stream $(1,6,120, \ldots) = (1!, 3!, 5!, \ldots)$. But $c$ does not represent
 	an infinite stream, since $\sigma(c)$ is undefined. 
\end{example}
The case of \emph{labelled transition systems} is a bit more subtle. The problem is that $(\powf \Sigma^* \emptyset)^A$
and $(\powc \Sigma^* \emptyset)^A$ are not DCPOs, in general, whereas the functor $(\pow -)^A$
does not have a cofree comonad. However, if the set of closed terms $\Sigma^* \emptyset$ is countable, then $(\powc\Sigma^* \emptyset)^A$
is a pointed DCPO, and thus Corollary~\ref{cor:model} applies. 
The specification in Example~\ref{ex:more-models} can be viewed as a specification for the functor
$(\powc -)^A$, and it has a countable set of terms. Therefore it has, by the Corollary,
a least supported model. In this model, the behaviour of $\sigma(t)$
is empty, for any $t \in \Sigma^* \emptyset$.


\section{Distributive laws for biGSOS specifications}\label{sec:dl}

In the previous section we have seen how to construct a least supported model of a monotone biGSOS specification,
as the least fixed point of a monotone function. 
In the present section 
we show that, given a monotone biGSOS specification, the construction of a least model generalizes to a \emph{lifting} of 
the free monad $\Sigma^*$ to the category of $B$-coalgebras. It then immediately follows that there exists a canonical
distributive law of the monad $\Sigma^*$ over the comonad $B^\infty$, and that the (unique) operational model of this
distributive law corresponds to the least supported model as constructed above. 

In order to proceed we define a \emph{$\DCPO$-ordered functor} as an ordered functor (Section~\ref{sec:similarity}) where $\PreOrd$ is 
replaced by $\DCPO$. Below we assume that $(B, \sqsubseteq)$ is $\DCPO$-ordered, and $\Sigma$
and $B$ are as before (having a free monad and cofree comonad respectively).
\begin{example}
A general class of functors that are $\DCPO$-ordered are those of the form $B + 1$, where 
the singleton $1$ is interpreted as the least element and all other distinct elements are incomparable (see Example~\ref{ex:ordered-functors-part}). Another example is the functor $(\pow -)^A$ of labelled transition systems
with arbitrary branching, but this example can not be treated here 
because there exists no cofree comonad for it. The case of labelled transition systems is 
treated in Section~\ref{sec:cpo}.
\end{example}
Let $\coalg{B}_{\Sigma^* X}$ be the set of $B$-coalgebras with carrier $\Sigma^* X$,
pointwise ordered as a DCPO by the order on $B$. 
The lifting of $\Sigma^*$ to $\coalg{B}$ that we are about to define maps a coalgebra $c \colon X \rightarrow BX$
to the least coalgebra $\liftco{c} \colon \Sigma^*X \rightarrow B\Sigma^*X$, w.r.t.\ the above
order on $\coalg{B}_{\Sigma^*X}$, making the following diagram commute. 
$$
\xymatrix{
	\Sigma B^\infty \Sigma^*X \ar[r]^{\rho_{\Sigma^*X}}
		& B\Sigma^* \Sigma^* X \ar[r]^{B\mu_X}
		& B\Sigma^*X
		& BX \ar[l]_{B\eta_X} \\
	\Sigma \Sigma^*X \ar[rr]_{\iota_X} \ar[u]^{\Sigma (\liftco{c})^\infty} 
		&  
		& \Sigma^*X \ar[u]^{\lift{c}}   
		& X \ar[l]^{\eta_X} \ar[u]_c
}
$$
Equivalently, $\liftco{c}$ is the least fixed point
of the operator $\varphi_c \colon \coalg{B}_{{\Sigma^* X}} \rightarrow \coalg{B}_{{\Sigma^* X}}$
defined by
$$
\varphi_c(f) = [B\mu_X \circ \rho_{\Sigma^*\emptyset} \circ \Sigma f^\infty, B\eta_X \circ c] \circ [\iota_X, \eta_X]^{-1} \,.
$$
Following the proof of Lemma~\ref{lm:phi-mon} it is 
easy to verify:
\begin{lemma}
	For any $c \colon X \rightarrow BX$, the function $\varphi_c$ is monotone. 
\end{lemma}

For the lifting of $\Sigma^*$,
we need to show that the above construction preserves coalgebra morphisms.

\begin{theorem}\label{thm:hom-pres}
	The functor $\lift{\Sigma^*} \colon \coalg{B} \rightarrow \coalg{B}$ defined by
	$$
		\lift{\Sigma^*}(X, c) = (\Sigma^*X, \liftco{c}) \qquad \text{ and } \qquad \lift{\Sigma^*}(h) = \Sigma^*h
	$$
	is a lifting of the functor $\Sigma^*$.
\end{theorem}
\begin{proof}
	Let $(X,c)$ and $(Y,d)$ be $B\Sigma^*$-coalgebras. We need to prove that, if $h \colon X \rightarrow Y$ is a coalgebra
	homomorphism from $c$ to $d$, then $\Sigma^*h$ is a homomorphism from $\liftco{c}$ to $\liftco{d}$.
	
The proof is by transfinite induction on the iterative
construction of $\liftco{c}$ and $\liftco{d}$ as limits of the ordinal-indexed initial chains of $\varphi_c$
and $\varphi_d$ respectively. 
For the limit (and base) case, given a (possibly empty) directed family of coalgebras $f_i \colon \Sigma^* X \rightarrow B\Sigma^*X$
and another directed family $g_i \colon \Sigma^* Y \rightarrow B\Sigma^*Y$, such that $B \Sigma^* h \circ f_i = g_i \circ \Sigma^* h$ for all $i$, we have
$B\Sigma^* h \circ \bigvee_i f_i = \bigvee_i (B\Sigma^* h \circ f_i) = \bigvee_i (g_i \circ \Sigma^* h) = (\bigvee_i g_i) \circ \Sigma^* h$ by continuity of $B\Sigma^* h$ and assumption.

Let $f \colon \Sigma^*X \rightarrow B\Sigma^*X$ and $g \colon \Sigma^* Y \rightarrow B\Sigma^* Y$ be such that $B\Sigma^*h \circ f = g \circ \Sigma^* h$. To prove:
$B\Sigma^*h \circ \varphi_c(f) = \varphi_d(g) \circ \Sigma^*h$, i.e., commutativity of the outside of:
$$
 \xymatrix@C=1.2cm{
 \Sigma^*X \ar[r]^-{[\iota_X,\eta_X]^{-1}} \ar[d]_{\Sigma^* h}  
	& \Sigma \Sigma^*X + X \ar[r]^-{\Sigma f^\infty + c} \ar[d]^{\Sigma \Sigma^* h + h} 
	& \Sigma B^\infty \Sigma^* X + BX \ar[r]^-{\rho_{\Sigma^*X} + \id} \ar[d]^{\Sigma B^\infty \Sigma^* h + Bh} 
	& B\Sigma^* \Sigma^* X + BX \ar[r]^-{[B \mu_X, B\eta_X]} \ar[d]^{B\Sigma^*\Sigma^*h + Bh} 
	& B\Sigma^* X \ar[d]^{B\Sigma^*h} \\
 \Sigma^*Y \ar[r]_-{[\iota_Y,\eta_Y]^{-1}} 
 	& \Sigma \Sigma^*Y + Y \ar[r]_-{\Sigma g^\infty + d} 
 	& \Sigma B^\infty \Sigma^* Y + BY \ar[r]_-{\rho_{\Sigma^*Y} + \id}
	& B\Sigma^* \Sigma^* Y + BY \ar[r]_-{[B \mu_Y, B\eta_Y]} 
	& B\Sigma^* Y
}
$$
From left to right, the first square commutes by naturality of $[\iota,\eta]$ (and the fact that 
it is an isomorphism), the second by assumption that $\Sigma^*h$ is a $B$-coalgebra homomorphism from $f$ to $g$ (and 
therefore a $B^\infty$-coalgebra homomorphism) and the assumption that $h$ is a coalgebra homomorphism from $c$ to $d$, 
the third by naturality of $\rho$, and the fourth by naturality of $\mu$ and $\eta$.
\end{proof}

We show that the (free) \emph{monad} $(\Sigma^*,\eta,\mu)$ lifts to $\coalg{B}$.
This is the heart of the matter. The main proof obligation is to
show that $\mu_X$ is a coalgebra homomorphism from $\lift{\Sigma^*}(\lift{\Sigma^*}(X,c))$
to $\lift{\Sigma^*}(X,c)$, for any $B$-coalgebra $(X,c)$. 

\begin{theorem}\label{thm:mu-hom}
	The \emph{monad} $(\Sigma^*,\eta,\mu)$ on $\Set$ lifts to the monad 
	$(\lift{\Sigma^*},\eta,\mu)$ on $\coalg{B}$, if $B$ preserves weak pullbacks. 
\end{theorem}
The lifting gives rise to a distributive law of monad over comonad.
\begin{theorem}\label{thm:dl}
	Let $\rho \colon \Sigma B^\infty \Rightarrow B \Sigma^*$ be a monotone biGSOS specification, where
	$B$ is $\DCPO$-ordered and preserves weak pullbacks. There exists a distributive law
	$\lambda \colon \Sigma^* B^\infty \Rightarrow B^\infty \Sigma^*$ of the free monad
	$\Sigma^*$ over the cofree comonad $B^\infty$ such that
	the operational model of $\lambda$ is the least supported model of $\rho$.
\end{theorem}
\begin{proof}
	By Theorem~\ref{thm:mu-hom}, we obtain a lifting of $(\Sigma^*,\eta,\mu)$
	to $\coalg{B}$. As explained in the preliminaries, such a lifting
	corresponds uniquely to a distributive law of the desired type. 
	The operational model of $\lambda$ is obtained
	by applying the lifting to the unique coalgebra $! \colon \emptyset \rightarrow B\emptyset$. 
	But that coincides, by definition of the lifting, with the least supported
	model as defined in Section~\ref{sec:mon-spec}.
\end{proof}

It follows from the general theory of bialgebras that the unique coalgebra morphism 
from the least supported model to the final coalgebra is an algebra homomorphism, i.e.,
behavioural equivalence on the least supported model of a monotone biGSOS specification is a congruence.

\paragraph{Labelled transition systems}
The results above do not apply to labelled transition systems. The problem is that
the cofree comonad for the functor $(\pow -)^A$ does not exist. 
A first attempt would be to restrict to the finitely branching transition systems, i.e., 
coalgebras for the functor $(\powf -)^A$. But this functor is not $\DCPO$-ordered, and indeed, 
contrary to the case of GSOS and coGSOS, 
even with a finite biGSOS specification one can easily generate a least model with infinite branching, so
that a lifting as in the previous section can not exist. 
\begin{example}
	Consider the following specification on (finitely branching) labelled transition systems, involving a unary operator $\sigma$ and a constant $c$:
	$$
	\frac{}{c \xrightarrow{a} \sigma(c)}
	\qquad 
	\frac{}{\sigma(x) \xrightarrow{a} \sigma(\sigma(x))}
	\qquad
	\frac{x \xrightarrow{a} x' \xrightarrow{a} x'' \xrightarrow{a} x'''}{\sigma(x) \xrightarrow{a} x'''}
	$$
	The left rule for $\sigma$ constructs an infinite chain of transitions from $\sigma(x)$ for any $x$,
	so in particular for $\sigma(c)$. The right rule takes the transitive closure of transitions from 
	$\sigma(c)$, so in the least model there are infinitely many transitions from $\sigma(c)$.
\end{example}
The model in the above example has countable branching. One might ask whether it can be adapted
to generate uncountable branching, i.e., that we can 
construct a biGSOS specification for the functor $(\powc -)^A$, such that the model
of this specification would feature uncountable branching. However, as it turns out, this is not the case,
at least if we assume $\Sigma$ to be a polynomial functor (a countable coproduct of finite products,
modelling a signature with countably many operations each of finite arity),
and the set of labels $A$ to be countable. This is shown more generally in the next section.

\section{Liftings for countably accessible functors}\label{sec:cpo} 

In the previous section, we have seen that one of the most important instances
of the framework---the case of labelled transition systems---does not work,
because of size issues: the functors in question either do
not have a cofree comonad, or are not DCPO-ordered. In the current section,
we solve this problem by showing that, if both functors $B, \Sigma$ are reasonably well-behaved,
then it suffices to have a DCPO-ordering of $B$ only on \emph{countable} sets.

More precisely, let $\Setc$ be the full subcategory of countable sets, 
with inclusion $I \colon \Setc \rightarrow \Set$. 
We assume that $(B,\sqsubseteq)$ is an ordered functor on $\Set$,
and that its restriction to countable sets is $\DCPO$-ordered:
$$
\xymatrix{
		& \DCPO \ar[r] & \PreOrd \ar[d]\\
	\Setc \ar[ur]^{\sqsubseteq} \ar[r]_{I} & \Set \ar[ur]^{\sqsubseteq} \ar[r]_{B}  & \Set
}
$$
This is a weaker assumption than in Section~\ref{sec:dl}:
before, every set $BX$ was assumed to be a pointed DCPO, whereas
here, they only need to be pointed DCPOs when $X$ is countable (and just a preorder otherwise).

\begin{example}\label{ex:restr-ord}
The functor $(\powc -)^A$ coincides
with the $\DCPO$-ordered functor $(\pow -)^A$ when restricted to countable sets, hence
it satisfies the above assumption. Notice that $(\powc-)^A$ is not $\DCPO$-ordered.
The functor $(\powf -)^A$ does not satisfy the above assumption.

The functor $(\M -)^A$, for the complete monoid $\real^+ \cup \{\infty\}$ (Example~\ref{ex:comonad:wts}), 
is ordered as a complete lattice~\cite{RotB16}, so also $\DCPO$-ordered. Similar to the above, the functor
$(\M_c -)^A$ is $\DCPO$-ordered when restricted to countable sets, i.e., satisfies the above assumption. 
\end{example}

We define $\coalgc{B}$ to be the full subcategory of $B$-coalgebras
whose carrier is a countable set, with inclusion $\lift{I} \colon \coalgc{B} \rightarrow \coalg{B}$. 
The associated forgetful functor is denoted by $U \colon \coalgc{B} \rightarrow \Setc$. 

The pointed DCPO structure on each $BX$, for $X$ countable, suffices to carry
out the fixed point constructions from the previous sections for coalgebras
over countable sets, if we assume that $\Sigma^*$ preserves countable sets. 
Notice, moreover, that the (partial) order 
on the functor $B$ is still necessary to define the simulation order on $B^\infty X$,
and hence speak about monotonicity of biGSOS specifications. The proof
of the following theorem is essentially the same as in the previous section. 

\begin{theorem}\label{thm:lifting-cnt}
Suppose $\Sigma^*$ preserves countable sets, and $B$ is an ordered
functor which preserves weak pullbacks and whose restriction to $\Setc$ is $\DCPO$-ordered. 
Let 
$(\Sigma^*_c,\eta^c,\mu^c)$ be the restriction of $(\Sigma^*,\eta,\mu)$ to $\Setc$.
Any monotone biGSOS specification $\rho \colon \Sigma B^\infty \Rightarrow B\Sigma^*$
gives rise to a lifting $(\lift{\Sigma^*}_c, \lift{\eta}^c, \lift{\mu}^c)$ of the monad
$(\Sigma^*_c,\eta^c,\mu^c)$
to $\coalgc{B}$. 
\end{theorem}

In the remainder of this section, we will show that, under certain
assumptions on $B$ and $\Sigma^*$, the above lifting extends
to a lifting of the monad $\Sigma^*$ from
$\Set$ to $\coalg{B}$, and hence a distributive
law of the monad $\Sigma^*$ over the cofree comonad $B^\infty$. 
It relies on the fact that, under certain conditions,
we can present every coalgebra as a (filtered) colimit of coalgebras over countable sets.

We use the theory of locally (countably, i.e., $\omega_1$-) presentable categories
and (countably) accessible categories. Because of space limits we can
not properly recall that theory in detail here (see~\cite{AR94}); 
we only recall a concrete characterisation of when a functor on $\Set$ is countably accessible,
since that will be assumed both for $B$ and $\Sigma^*$ later on.
On $\Set$, a functor $B \colon \Set \rightarrow \Set$ is 
\emph{countably accessible} if for every set $X$ and element
$x \in BX$, there is an injective function $i \colon Y \rightarrow X$
from a finite set $Y$ and an element $y \in BY$ such that $Bi(y) = x$.
Intuitively, such functors are determined by how they operate on countable sets.
\begin{example}\label{ex:count-acc}
	Any finitary functor is countably accessible. Further, 
	the functors $(\powc-)^A$ and $(\M_c -)^A$ (c.f.\ Example~\ref{ex:restr-ord}) 
	are countably accessible if $A$ is countable. 
\end{example}
A functor is called \emph{strongly countably accessible} if
it is countably accessible and additionally preserves countable sets, i.e.,
it restricts to a functor $\Setc \rightarrow \Setc$. We will assume
this for our ``syntax'' functor $\Sigma^*$. If $\Sigma$ correponds
to a signature with countably many operations each of finite arity (so is a countable coproduct of finite products)
then $\Sigma^*$ is strongly countably accessible. 

The central idea of obtaining a lifting to $\coalg{B}$ from a lifting to $\coalgc{B}$
is to \emph{extend} the monad on $\coalgc{B}$ along the inclusion $\lift{I} \colon \coalgc{B} \rightarrow \coalg{B}$.
Concretely, a functor $T \colon \Set \rightarrow \Set$ extends $T_c \colon \Setc \rightarrow \Setc$
if there is a natural isomorphism $\alpha \colon IT_c \Rightarrow TI$. A monad
$(T,\eta,\mu)$ on $\Set$ extends a monad $(T_c, \eta_c, \mu_c)$ on $\Setc$
if $T_c$ extends $T$ with some isomorphism $\alpha$ such that 
$\alpha \circ I\eta_c = \eta I$ 
and $\alpha \circ I \mu_c= \mu I \circ T\alpha \circ \alpha T_c$.
This notion of extension is generalised naturally to arbitrary locally countably
presentable categories. Monads on the category of countably presentable objects can always be extended.
 \begin{lemma}\label{lm:ext-monad-basic}	
 	Let $\C$ be a locally countably presentable category, 
 	with $I \colon \C_\cnt \rightarrow \C$ the subcategory of countably
 	presentable objects.
	Any monad $(T_\cnt, \eta^\cnt,\mu^\cnt)$ on $\C_\cnt$
	extends uniquely to a monad $(T,\eta,\mu)$ on $\C$, along $I \colon \C_\cnt \rightarrow \C$.
\end{lemma}
Since $B$ is countably accessible,
$\coalg{B}$ is locally countably presentable and $\coalgc{B}$ is
the associated category of countably presentable objects~\cite{AdamekP04}.
This means every $B$-coalgebra can be presented
as a filtered colimit of $B$-coalgebras with countable carriers.
The above lemma applies, so we can extend the monad on $\coalgc{B}$
of Theorem~\ref{thm:lifting-cnt} to a monad on $\coalg{B}$, 
resulting in Theorem~\ref{thm:dl2} below. 
The latter relies on Theorem~\ref{thm:cube}, which ensures that, doing so, we will get a lifting 
of the monad on $\Set$ that we started with. 

In the remainder of this section, we will consider a slightly 
relaxed version of functor liftings, up to isomorphism,
similar to extensions defined before. This is harmless---those
still correspond to distributive laws---but since
the monad on $\coalg{B}$ is constructed only up to isomorphism,
it is more natural to work with in this setting. 
We say $(\lift{T}, \lift{\eta}, \lift{\mu})$ lifts
$(T,\eta,\mu)$ (up to isomorphism) if there is a natural isomorphism $\alpha \colon U\lift{T} \Rightarrow TU$
such that $\alpha \circ U\lift{\eta} = \eta U$ and $\alpha \circ U \lift{\mu}= \mu U \circ T\alpha \circ \alpha \lift{T}$.

 \begin{theorem}\label{thm:cube}
	Let $B \colon \Set \rightarrow \Set$ be countably accessible. 
	Suppose $(T_\cnt,\eta^\cnt,\mu^\cnt)$ is a monad on $\Setc$, which
	lifts to a monad 
	$(\lift{T}_\cnt, \lift{\eta}^\cnt, \lift{\mu}^\cnt)$
	on $\coalgc{B}$.
	Then 
	\begin{enumerate}
		\item $(T_\cnt,\eta^\cnt,\mu^\cnt)$ extends to $(T,\eta,\mu)$ along $I \colon \Set_\cnt \rightarrow \Set$,
		\item $(\lift{T}_\cnt, \lift{\eta}^\cnt, \lift{\mu}^\cnt)$ extends to $(\lift{T}, \lift{\eta}, \lift{\mu})$
		along $\lift{I} \colon \coalgc{B} \rightarrow \coalg{B}$,
		\item $(\lift{T}, \lift{\eta}, \lift{\mu})$ is a lifting (up to isomorphism) of $(T,\eta,\mu)$.
	\end{enumerate}
\end{theorem}
By instantiating the above theorem with the lifting of Theorem~\ref{thm:lifting-cnt}, the third point 
gives us the desired lifting to $\coalg{B}$. In particular $T_\cnt$ is instantiated
to the restriction $\Sigma^*_\cnt$ of $\Sigma^*$, which means that the extension
in the first point is just $\Sigma^*$ itself. 
 \begin{theorem}\label{thm:dl2}
	Let $\rho \colon \Sigma B^\infty \Rightarrow B \Sigma^*$ be a monotone biGSOS specification, where
	$B$ is an ordered functor whose restriction to countable sets is $\DCPO$-ordered, $B$ is countably accessible,
	$B$ preserves weak pullbacks,
	and $\Sigma^*$ is strongly countably accessible. There exists a distributive law
	$\lambda \colon \Sigma^* B^\infty \Rightarrow B^\infty \Sigma^*$ of the free monad
	$\Sigma^*$ over the cofree comonad $B^\infty$ such that
	the operational model of $\lambda$ is the least supported model of $\rho$.
\end{theorem}
As explained in Example~\ref{ex:count-acc} and Example~\ref{ex:restr-ord},
if $B$ is either $(\pow_c -)^A$ or $(\M_c -)^A$ (weighted in the non-negative real numbers) with $A$ countable,
then it satisfies the above hypotheses (that $\M_c$ preserves weak pullbacks follows essentially from~\cite{GummS01}). 
So the above theorem applies to labelled transition systems and weighted transition systems (of the above type) over a countable
set of labels, as long as the syntax is composed of countably many operations each with finite arity. 
Hence, behavioural equivalence on the operational model of any biGSOS specification for such systems is a congruence.

\section{Future work}

In this paper we provided a bialgebraic foundation of positive specification formats over ordered functors, 
involving rules that feature lookahead in the premises as well as complex terms in conclusions. From a practical point of view,
it would be interesting to find more concrete rules formats corresponding to the abstract format of the present paper. In particular, concrete GSOS formats for weighted transition systems exist~\cite{Klin09}; they 
could be a good starting point.

It is currently unclear to us whether the assumption of weak pullback preservation in the main results is necessary.
This assumption is used in our proof of Lemma~\ref{lm:ineq-comonad}, which in turn is used in the proof
that the free monad lifts to the category of coalgebras (Theorem~\ref{thm:mu-hom}).
Finally, we would like to study \emph{continuous} specifications, as opposed to specifications that 
are only monotone, as in the current paper. Continuous specifications should be better
behaved than monotone ones. However, it is currently not yet clear how to characterize
continuity of a specification both at the concrete, syntactic level. 

\bibliographystyle{eptcs}

\bibliography{monotone}

\end{document}